\documentclass[preprint,12pt]{elsarticle}

\usepackage{yhmath}
\usepackage{url}
\usepackage{amssymb,wasysym}
\usepackage{graphicx}
\usepackage{wrapfig}
\usepackage{subfig}
\usepackage{color}
\usepackage{multirow}
\usepackage{epstopdf}
\usepackage[colorlinks=true]{hyperref}
\usepackage[usenames,dvipsnames]{xcolor,colortbl}
\usepackage[margin=1.2in]{geometry} 
\usepackage{lscape}

\def\rknn{\mbox{R$k$NN}}
\def\r1nn{\mbox{R$1$NN}}
\def\knng{\mbox{$k$-NNG}}
\def\ksyg{\mbox{$k$-SYG}}
\def\1syg{\mbox{$1$-SYG}}

\newcommand{\etal}{\emph{et~al.}}
\newcommand{\ie}{\emph{i.e.}}
\newcommand{\eg}{\emph{e.g.}}

\graphicspath{{Images/}}

\newtheorem{theorem}{Theorem}[section]
\newtheorem{lemma}{Lemma}[section]
\newproof{proof}{Proof}

\newtheorem{corollary}{Corollary}[section]

\journal{Comput. Geom. Theory Appl.}

\begin{document}
\begin{frontmatter}

\title{Kinetic $k$-Semi-Yao Graph and its Applications~\tnoteref{t1}}


\author{Zahed Rahmati\corref{cor1}}
\ead{zahedrahmati@gmail.com}

\author[rvt2]{Mohammad Ali Abam}
\ead{abam@sharif.edu}

\author[rvt1]{Valerie King}
\ead{val@uvic.ca}

\author[rvt1]{Sue Whitesides}
\ead{sue@uvic.ca}

\tnotetext[t1]{Preliminary versions of parts of this paper appeared in \textit{Proceedings of the 26th Canadian Conference on Computational Geometry} (CCCG 2014)~\cite{DBLP:journals/corr/RahmatiA13} and \textit{Proceedings of the 25th International Workshop on Combinatorial Algorithms} (IWOCA 2014)~\cite{RahmatiKW14_IWOCA14}.}

\cortext[cor1]{Corresponding author}

\address[rvt2]{Dept. of Computer Engineering, Sharif University of Technology, Tehran, Iran.}
\address[rvt1]{Dept. of Computer Science, University of Victoria, Victoria, Canada.}

\begin{abstract}
This paper introduces a new proximity graph, called \textit{the $k$-Semi-Yao graph} (\ksyg), on a set $P$ of points in $\mathbb{R}^d$, which is a supergraph of the $k$-nearest neighbor graph (\knng) of $P$. We provide a kinetic data structure (KDS) to maintain the \ksyg~on moving points, where the trajectory of each point is a polynomial function whose degree is bounded by some constant. Our technique gives the first KDS for the theta graph (\ie, \1syg) in $\mathbb{R}^d$. It generalizes and improves on previous work on maintaining the theta graph in~$\mathbb{R}^2$. 

As an application, we use the kinetic \ksyg~to provide the \textit{first} KDS for maintenance of all the $k$-nearest neighbors in $\mathbb{R}^d$, for any $k\geq 1$. Previous works considered the $k=1$ case only.

Our KDS for all the $1$-nearest neighbors is \textit{deterministic}. The best previous KDS for all the $1$-nearest neighbors in $ \mathbb{R}^d$ is \textit{randomized}. Our structure and analysis are simpler and improve on this work for the $k=1$ case. We also provide a KDS for all the $(1+\epsilon)$-nearest neighbors, which in fact gives better performance than previous KDS's for maintenance of all the exact $1$-nearest neighbors.

As another application, we present the \textit{first} KDS for answering \textit{reverse} $k$-nearest neighbor queries on moving points in $ \mathbb{R}^d$, for any $k\geq 1$.
\end{abstract}

\end{frontmatter}
\section{Introduction}
The physical and virtual worlds around us are full of moving objects, including players in multi-player game environments, soldiers in a battlefield, tourists in foreign environments, and mobile devices in wireless ad-hoc networks. The problems that deal with attributes  (\eg, closest pair) of sets of objects arising from the distances between objects are known as \textit{proximity problems}.  Considering (kinetic version of) a proximity problem on moving objects in order to solve a proximity problem is called a \textit{kinetic proximity problem}.

The maintenance of attributes of sets of moving points has been studied extensively over the past 15 years; see~\cite{rahmati2014simple} and references therein. A basic framework for this study, which is described in Section~\ref{sec:KDSframework}, is that of a kinetic data structure (KDS) which is in fact a set of data structures and algorithms to track the attributes of moving points. We consider some  fundamental proximity problems, which are stated in Section~\ref{sec:ProbStatment}, in this standard KDS model.

\subsection{Problem Statement}\label{sec:ProbStatment}
Let $P$ be a set of $n$ points in $\mathbb{R}^d$, where $d$ is arbitrary but fixed. Finding the $k$-nearest neighbors to a query point, which is called the \textit{$k$-nearest neighbor} problem, is fundamental in computational geometry. The \textit{all $k$-nearest neighbors} problem, a variant of the $k$-nearest neighbor problem, is to find the $k$-nearest neighbors to each point $p\in P$. Given any $\epsilon>0$, the \textit{all $(1+\epsilon)$-nearest neighbors} problem is to find some $\hat{q}\in P$ for each point $p\in P$, such that the Euclidean distance $|p\hat{q}|$ between $p$ and $\hat{q}$ is within a factor of $(1+\epsilon)$ of the Euclidean distance between $p$ and its nearest neighbor. The graph constructed by connecting each point $p\in P$ to its $k$-nearest neighbors is called the \textit{$k$-nearest neighbor graph} (\knng). The \textit{closest pair} problem is to find the endpoints of the edge in the $1$-NNG whose separation distance is minimum. 
The \textit{theta graph} is a well-studied sparse proximity graph~\cite{Clarkson:1987:AAS:28395.28402,Keil:1988:ACE:61764.61787}. This graph is constructed as follows. Partition the space around each point $p\in P$ into $c$ polyhedral cones $C_l(p)$, $0\leq l \leq c-1$. In each cone $C_l(p)$, a vector $x_l(p)$ is chosen as the cone axis. Then connect the point $p$ to a particular point inside each cone $C_l(p)$, where the particular point is the element of $P$ with minimum length projection on $x_l(p)$~\footnote{By treating $c$ as a parameter of the theta graph, one can obtain an important class of sparse graphs, called \textit{t-spanners}, with different stretch factors $t$~\cite{DBLP:journals/corr/BoseCMRV14}.}.

The \textit{reverse $k$-nearest neighbor} (\rknn) problem is a variant of the $k$-nearest neighbor problem that asks for the influence of a query point on a point set $P$. Unlike the $k$-nearest neighbor problem, the exact number of reverse $k$-nearest neighbors of a query point is not known in advance, but as we prove in Lemma~\ref{the:RkNNsUprBnd} the number is upper-bounded by $O(k)$. The \rknn~problem is formally defined as follows: Given a query point $q\notin P$, find the set $\rknn(q)$ of all $p$ in $P$ for which $q$ is one of $k$-nearest neighbors of $p$. Thus $\rknn(q)=\{p\in P:~|pq|\leq |pp_k|\}$, where $|.|$ denotes Euclidean distance, and $p_k$ is the $k^{th}$ nearest neighbor of $p$ among the points in $P$.

\subsection{KDS Framework}\label{sec:KDSframework}
Basch, Guibas, and Hershberger~\cite{Basch:1997:DSM:314161.314435} introduced the \textit{kinetic data structure framework} to maintain attributes (\eg, closest pair) of moving points. In the kinetic setting, we assume each coordinate of the trajectory of a point $p\in P$ is a polynomial function of degree bounded by some constant $s$. The correctness of an attribute over time is determined based on correctness of a set of \textit{certificates}. A certificate is a boolean function of time, and its \textit{failure time} is the next time after the current time at which the certificate will become invalid. When a certificate fails, we say that an \textit{event} occurs. Using a \textit{priority queue} of the failure times of the certificates, we can know the next time after the current time that an event occurs. When the failure time of the certificate with highest priority in the priority queue is equal to the current time we invoke the update mechanism to reorganize the data structures and replace the invalid certificates with new valid ones. 

To analyse the performance of a KDS there are four standard criteria. A KDS distinguishes between two types of events: \textit{external events} and \textit{internal events}. An event that changes the desired attribute itself is called an external event, and those events that cause only some internal changes in the data structures are called internal events. If the ratio between the  worst-case number of internal events in the KDS to the worst-case number of external events is $O(\text{polylog}(n))$, the KDS is \textit{efficient}. If the response time of the update mechanism to an event is $O(\text{polylog}(n))$, the KDS is \textit{responsive}. The compactness of a KDS refers to size of the priority queue at any fixed time: if the KDS uses $O(n.\text{polylog}(n))$ certificates, it is \textit{compact}. The KDS is \textit{local} if the number of certificates associated with any point at any fixed time is $O(\text{polylog}(n))$. The locality of a KDS is an important criterion; if a KDS is local, it can be updated quickly when a point changes its trajectory.

\subsection{Related Work}\label{sec:RelatedWork}
\paragraph{Stationary setting.} 
For a set $P$ of $n$ stationary points, the closest pair problem can be solved in $O(n\log n)$ time~\cite{4567872,Bentley:1976:DMS:800113.803652}. There is also a linear-time randomized algorithm to find the closest pair~\cite{FSHJ1078}. One can report all the $1$-nearest neighbors in time $O(n\log n)$~\cite{Vaidya:1989:ONL:70530.70532}. For any $k\geq 1$, all the $k$-nearest neighbors can be reported in time $O(kn\log n)$~\cite{Dickerson:1996:APP:236464.236474}, in order of increasing distance; reporting the unordered set takes time $O(n\log n + kn)$~\cite{Callahan366854,Clarkson:1983:FAN:1382437.1382825,Dickerson:1996:APP:236464.236474}. 

The reverse $k$-nearest neighbor problem was first posed by Korn and Muthukrishnan~\cite{Korn:2000:ISB:342009.335415}  in the database community, where it was then considered extensively due to its many applications in, for example, decision support systems, profile-based marketing, traffic networks, business location planning, clustering and outlier detection, and molecular biology~\cite{Kumar:2008:EAR:1463434.1463483,DBLP:conf/sdm/LinED08}. In computational geometry, there exist two data structures~\cite{MaheshwariVZ02cccg2002,CheongIJCGA2011} that give solutions to the \rknn~problem. Both of these solutions only work for $k=1$. Maheshwari~\etal~\cite{MaheshwariVZ02cccg2002} gave a data structure to solve the \r1nn~problem in $\mathbb{R}^2$. Their data structure uses $O(n)$ space and $O(n\log n)$ preprocessing time, and an \r1nn~query can be answered in time $O(\log n)$. Cheong~\etal~\cite{CheongIJCGA2011} considered the \r1nn~problem in $\mathbb{R}^d$, where $d=O(1)$. Their method gives the same complexity as that of~\cite{MaheshwariVZ02cccg2002}~\footnote{It seems that the approach by Cheong~\etal~can be extended to answer \rknn~queries, for any $k\geq 1$, with preprocessing time $O(kn\log n)$, space $O(kn)$, and query time $O(\log n + k)$.}.

\paragraph{Kinetic setting.}
For a set of $n$ moving points in $\mathbb{R}^2$, where each trajectory of a point is a polynomial function of degree bounded by constant $s$, Basch, Guibas, and Hershberger~\cite{Basch:1997:DSM:314161.314435} provided a KDS for maintenance of the closest pair. Their KDS uses linear space and processes $O(n^2\beta_{2s+2}(n)\log n)$ events, each in time $O(\log^2 n)$. Here, $\beta_s(n)$ is an extremely slow-growing function.


Basch, Guibas, and Zhang~\cite{Basch:1997:PPM:262839.262998} used multidimensional range trees to maintain the closest pair in $\mathbb{R}^d$. For a fixed dimension $d$, their KDS uses $O(n\log^{d-1}n)$ space and processes $O(n^2\beta_{2s+2}(n)\log n)$ events, each in worst-case time $O(\log^d n)$. Their KDS is responsive, efficient, compact, and local.

Using multidimensional range trees, Agarwal, Kaplan, and Sharir (TALG'08)~\cite{Agarwal:2008:KDD:1435375.1435379} gave KDS's both for maintenance of the closest pair and for all the $1$-nearest neighbors in $\mathbb{R}^d$. The closest pair KDS by Agarwal~\etal~uses $O(n\log^{d-1} n)$ space and processes $O(n^2\beta_{2s+2}(n)\log n)$ events, each in amortized time $O(\log^d n)$; this KDS is efficient, amortized responsive, local, and compact. Agarwal~\etal~gave the first efficient KDS to maintain all the $1$-nearest neighbors in $\mathbb{R}^d$. For the efficiency of their KDS, they implemented range trees by using randomized search trees (treaps). Their \textit{randomized} kinetic approach uses $O(n\log^d n)$ space and processes $O(n^2\beta_{2s+2}^2(n)\log^{d+1} n)$ events; the expected time to process all events is $O(n^2\beta_{2s+2}^2(n)\log^{d+2} n)$. Their all $1$-nearest neighbors KDS is efficient, amortized responsive, compact, but in general is not local.

Rahmati, King, and Whitesides~\cite{Rahmati2014} gave the first KDS for maintenance of the theta graph in $\mathbb{R}^2$. Their method uses a constant number of kinetic Delaunay triangulations to maintain the theta graph. Their theta graph KDS uses linear space and processes $O(n^2\beta_{2s+2}(n))$ events with total processing time $O(n^2\beta_{2s+2}(n)\log n)$. Using the kinetic theta graph, they improved the previous KDS by Agarwal~\etal~to maintain all the $1$-nearest neighbors in $\mathbb{R}^2$. In particular, their \textit{deterministic} kinetic algorithm, which is also arguably simpler than the randomized kinetic algorithm by Agarwal~\etal, uses $O(n)$ space and processes $O(n^2\beta_{2s+2}^2(n)\log n)$ events with total processing time $O(n^2\beta_{2s+2}^2(n)\log^2 n)$. With the same complexity as their KDS for maintenance of all the $1$-nearest neighbors, they maintain the closest pair over time. Their KDS's for maintenance of the theta graph, all the $1$-nearest neighbors, and the closest pair are efficient, amortized responsive, compact, but in general are not local.

The reverse $k$-nearest neighbor queries for a set of continuously moving objects has attracted the attention of the database community (see~\cite{Cheema:2012:CRK:2124885.2124903} and references therein). To our knowledge there is no previous solution to the kinetic \rknn~problem in the literature.
\subsection{Our Contributions}\label{sec:Contributions} 
We introduce a new sparse proximity graph, called the \textit{$k$-Semi-Yao graph} (\ksyg), and then maintain the \ksyg~for a set of $n$ moving points, where the trajectory of each point is a polynomial function of at most constant degree $s$. We use a constant number of range trees to apply necessary changes to the \ksyg~over time. We prove that the edge set of the \ksyg~includes the pairs of the $k$-nearest neighbors as a subset. This enables us to easily provide the \textit{first} kinetic solutions in $\mathbb{R}^d$ for maintenance of all the $k$-nearest neighbors, and then, as another first, to answer \rknn~queries on moving points, for any $k\geq 1$.

Our KDS for maintenance of the \1syg~(\ie, theta graph), in fixed dimension $d$, uses $O(n\log^d n)$ space and processes $O(n^2)$ events with total processing time $O(n^2\beta_{2s+2}(n)\log^{d+1} n)$.  The KDS is compact, efficient, amortized responsive, and it is local. Our KDS generalizes the previous KDS for the \1syg~by Rahmati~\etal~\cite{Rahmati2014} which only works in $\mathbb{R}^2$. Also, our kinetic approach yields improvements on the previous KDS for maintenance of the \1syg~by Rahmati~\etal~\cite{Rahmati2014}: Our KDS is local, but their KDS is not;  in particular, each point in our KDS participates in $O(1)$ certificates, but in their KDS each point participates in $O(n)$ certificates. Also, our KDS handles $O(n^2)$ events, but their KDS handles $O(n^2\beta_{2s+2}(n))$ events in $\mathbb{R}^2$.

Our KDS for maintenance of all the $1$-nearest neighbors uses $O(n\log^d n)$ space and processes $O(n^2\beta_{2s+2}^2(n)\log n)$) events; the total processing time to handle all the events is $O(n^2\beta_{2s+2}(n)\log^{d+1} n)$.  Our KDS is compact, efficient, amortized responsive, but it is not local in general. For each point $p\in P$ in the \1syg~we construct a tournament tree to maintain the edge with minimum length among the edges incident to the point $p$. Summing over elements of all the tournament trees in our KDS is linear in $n$, which leads to the total number of events $O(n^2\beta_{2s+2}^2(n)\log n)$, which is \textit{independent} of $d$. Our \textit{deterministic} method improves with simpler structure and analysis of the previous \textit{randomized} kinetic algorithm by Agarwal~\etal~\cite{Agarwal:2008:KDD:1435375.1435379}: The expected total size of the tournament trees in their KDS for all $1$-nearest neighbors is $O(n\log^dn)$; thus their KDS processes $O(n^2\beta_{2s+2}^2(n)\log^{d+1} n)$ events, which depends on $d$. Also, we improve their KDS by a factor of $\log n$ in the total cost. Furthermore, on average, each point in our KDS participates in $O(1)$ events, but in their KDS each point participates in $O(\log^d n)$ events.

For maintaining all the $1$-nearest neighbors, neither our KDS nor the KDS by Agarwal~\etal~is local in the \textit{worst-case}, and furthermore, each event in our KDS and in their KDS is handled in a polylogarithmic \textit{amortized} time. To satisfy the locality criterion and to get a worst-case processing time for handling events, we provide a KDS for all the $(1+\epsilon)$-nearest neighbors. In particular, for each point $p$ we maintain some point $\hat{q}$ such that $|p\hat{q}|<(1+\epsilon).|pq|$, where $q$ is the nearest neighbor of $p$ and $|pq|$ is the Euclidean distance between $p$ and $q$. This KDS uses $O(n\log^{d} n)$ space, and handles $O(n^2\log^d n)$ events, each in worst-case time $O(\log^d n\log\log n)$; it is compact, efficient, responsive, and local.


To answer an \rknn~query for a query point $q\notin P$ at any time $t$, we partition the $d$-dimensional space into a constant number of cones around $q$, and then among the points of $P$ in each cone, we examine the $k$ points having shortest projections on the cone axis. We obtain $O(k)$ candidate points for $q$ such that $q$ might be one of their $k$-nearest neighbors at time $t$. To check which if any of these candidate points is a reverse $k$-nearest neighbor of $q$, we maintain the $k^{th}$ nearest neighbor $p_k$ of each point $p\in P$ over time. By checking whether $|pq|\leq |pp_k|$ we can easily check whether a candidate point $p$ is one of the reverse $k$-nearest neighbors of $q$ at time $t$. Given a KDS for maintenance of all the $k$-nearest neighbors, an \rknn~query can be answered at any time $t$ in $O(\log^d n+k)$ time. Note that if an event occurs at the same time $t$, we first spend amortized time $O(\text{polylog}(n))$ to update all the $k$-nearest neighbors, and then we answer the query.

Table~\ref{tab:AllResults} summarizes all the (previous and new) results for the kinetic proximity problems.  In this table, ``Dim.", ``Num.", and ``Proc." stand for ``Dimension", ``Number", and ``Processing", respectively. Here, $\beta(n)$ is an extremely slow-growing function, and $\phi(n)$ is the complexity of the $k$-level, which are defined in Theorems~\ref{the:totallyDFcomplexity} and~\ref{the:k_levelComplexity}, respectively.

\begin{landscape}
\begin{table}
\centering
\begin{tabular}{|l|c|c|c|c|c|} \hline
\textsl{Kinetic problem} & \textsl{Dim.} & \textsl{Space} & \textsl{Num. of events} & \textsl{Proc. time} & \textsl{Local} \\ \hline
      Closest pair~\cite{Basch:1997:DSM:314161.314435}  & $d=2$  & $O(n)$  &   $O(n^2\beta(n)\log n)$  &   $O(\log^2 n)$  /event  &   Yes   \\ \hline
      Closest pair~\cite{Basch:1997:PPM:262839.262998} & $d=O(1)$ & $O(n\log^{d-1} n)$  &   $O(n^2\beta(n)\log n)$  &   $O(\log^d n)$ /event  &   Yes   \\ \hline
      Closest pair~\cite{Agarwal:2008:KDD:1435375.1435379} & $d=O(1)$  & $O(n\log^{d-1} n)$  &   $O(n^2\beta(n)\log n)$  &   $O(\log^d n)$ /event  &   Yes   \\ \hline
      Closest pair~\cite{Rahmati2014} & $d=2$ & $O(n)$  &   $O(n^2\beta^2(n)\log n)$  &   $O(n^2\beta^2(n)\log^2 n)$  &   No   \\ \hline            
      All $1$-NNs~\cite{Agarwal:2008:KDD:1435375.1435379} & $d=O(1)$  & $O(n\log^d n)$  &   $O(n^2\beta(n)\log^{d+1} n)$  &   $O(n^2\beta(n)\log^{d+2} n)$  &   No   \\ \hline
      All $1$-NNs~\cite{Rahmati2014} & $d=2$ & $O(n)$  &   $O(n^2\beta^2(n)\log n)$  &   $O(n^2\beta^2(n)\log^2 n)$  &   No   \\ \hline 
      All $1$-NNs~[Here] & $d=O(1)$  & $O(n\log^d n)$  &   $O(n^2\beta^2(n)\log n)$  &   $O(n^2\beta(n)\log^{d+1} n)$   &   No   \\ \hline       
      All $(1+\epsilon)$-NNs [Here] & $d=O(1)$  & $O(n\log^d n)$  &   $O(n^2\log^d n)$  &   $O(\log^d n\log\log n)$ /event &   Yes   \\ \hline   
     All $k$-NNs~[Here] & $d=O(1)$  & $O(n\log^{d+1} n+kn)$  &   $O(n^2\phi(n))$  &   $O(n^2\phi(n)\log n)$ &   No   \\ \hline
      \1syg~\cite{Rahmati2014} & $d=2$  & $O(n)$  &   $O(n^2\beta(n))$  &   $O(n^2\beta(n)\log n)$  &   No   \\ \hline 
      \1syg~[Here] & $d=O(1)$  & $O(n\log^d n)$  &   $O(n^2)$  &   $O(n^2\beta(n)\log^{d+1} n)$  &   Yes   \\ \hline
\end{tabular}
\caption{The previous results and our results for kinetic proximity problems.}
\label{tab:AllResults}
\end{table}
\end{landscape}
\subsection{Outline}\label{sec:Outline} 
In Section~\ref{sec:Preliminaries}, we describe the necessary background and review the theorems that we use throughout this paper. Section~\ref{sec:Relationship} provides key lemmas, and in fact introduces a new supergraph, namely the \textit{$k$-Semi-Yao graph} (\ksyg), of the \knng. Section~\ref{sec:ReportKNNs} shows how to construct the \ksyg~and report all the $k$-nearest neighbors. Section~\ref{sec:kineticKSYG} gives a kinetic approach for maintenance of the $\ksyg$. Section~\ref{sec:applications}  provides two applications of the kinetic $\ksyg$: maintenance of all the $k$-nearest neighbors, and answering \rknn~queries on moving points. Section~\ref{sec:KineticEpsANN} shows how to maintain all the $(1+\epsilon)$-nearest neighbors. Section~\ref{sec:conclusion} concludes.

\section{Preliminaries}\label{sec:Preliminaries}
\paragraph{Partitioning space around the origin.}
Let $\overrightarrow{v}$ be a unit vector in $\mathbb{R}^d$ with apex at the origin $o$, and let $\theta$ be a constant. We define an \textit{infinite right circular cone} with respect to $\overrightarrow{v}$ and $\theta$ to be the set of points $x\in \mathbb{R}^d$ such that the angle between $\overrightarrow{ox}$ and $\overrightarrow{v}$ is at most $\theta/2$; Figure~\ref{fig:circularCone}(a) depicts an infinite right circular cone in $\mathbb{R}^3$. We define a \textit{polyhedral cone} of opening angle $\theta$ with respect to $\overrightarrow{v}$ to be the intersection of $d$ half-spaces such that the intersection is contained in an infinite right circular cone with respect to $\overrightarrow{v}$ and $\theta$, and such that all the half-spaces contain the origin $o$; Figure~\ref{fig:circularCone}(b) depicts a polyhedral cone in $\mathbb{R}^3$, which is contained in the infinite right circular cone of Figure~\ref{fig:circularCone}(a). The angle between any two rays inside a polyhedral cone of opening angle $\theta$ emanating from $o$ is at most $\theta$.

\begin{figure}[t!]
\centering
\includegraphics[scale=1]{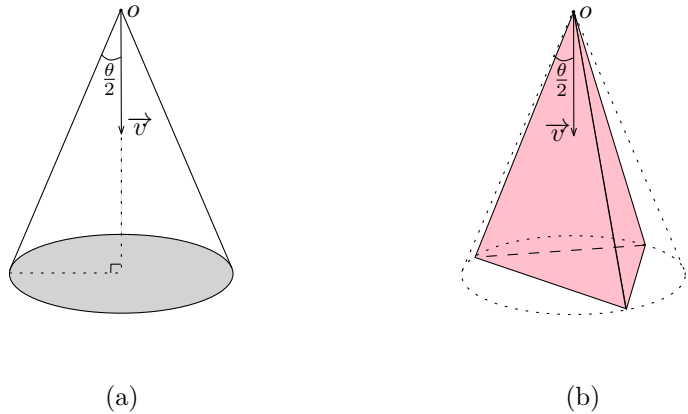}
\caption{An infinite right circular cone and a polyhedral cone.}
\label{fig:circularCone}
\end{figure}

\begin{lemma}{\tt \cite{Abam:2011:KSX:1971362.1971367}}\label{the:NumofPolyCones}
The $d$-dimensional space around a point can be covered by a collection of $c=O(1/\theta^{d-1})$ interior-disjoint polyhedral cones of opening angle $\theta$.
\end{lemma}

\paragraph{Kinetic rank-based range tree (RBRT).} Let $C=\{C_0,...,C_{c-1}\}$ be a set of polyhedral cones of opening angle $\theta$ with their apex at the origin $o$ that together cover $\mathbb{R}^d$. Denote by $f_1,...,f_d$ the bounding half-spaces $f^+_1,...,f^+_d$ of $C_l$, $0\leq l\leq c-1$. Let $u_i$ be the normal to $f^+_i$, $1\leq i\leq d$. Figure~\ref{fig:aConeC_l}(a) depicts $u_1$ and $u_2$ for the half-spaces $f^+_1$ and $f^+_2$ of a polyhedral cone $C_l\in C$ in $\mathbb{R}^2$. Let $C_l(p)$ denote the translated copy of $C_l$ with apex at $p$; see Figure~\ref{fig:aConeC_l}(b). 

\begin{figure}[t!]
\begin{center}
\includegraphics[scale=1]{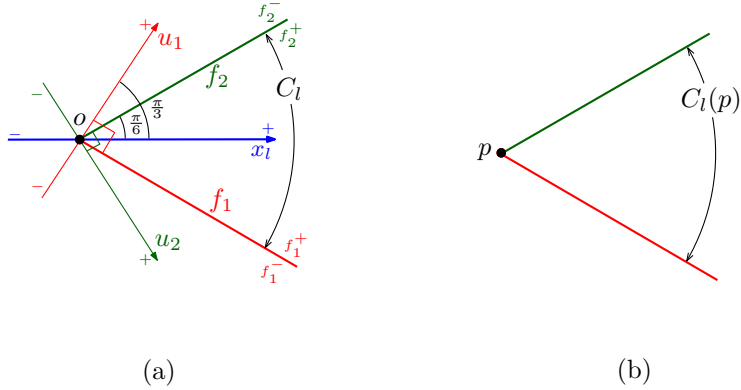}
\end{center}
\caption{(a) The cone $C_l$ in $\mathbb{R}^2$ is bounded by $f_1$ and $f_2$. The coordinate axes $u_1$ and $u_2$ are orthogonal to $f_1$ and $f_2$. (b) A translation of $C_l$ that moves the apex to $p$.}
\label{fig:aConeC_l}
\end{figure}

Consider a set $P$ of moving points. Using a \textit{kinetic range tree} data structure, one can process the moving points in $P$ such that the points of $P$ inside a query range can efficiently be reported at any time $t$. By creating a kinetic range tree data structure ${\cal T}_l$ for the polyhedral cone $C_l$, one can report the points in $P\cap C_l(p)$ for a query range $C_l(p)$ at time $t$. 

Abam and de Berg~\cite{Abam:2011:KSX:1971362.1971367} introduced a variant of range trees, a \textit{rank-based range tree} (RBRT), that avoids rebalancing the range tree and gives a polylogarithmic worst-case processing time when an event occurs. Similar to a regular range tree (see~\cite{Berg:2008:CGA:1370949}), the points at level $i$ of an RBRT ${\cal T}_l$, which is an RBRT corresponding to $C_l$, are sorted at the leaves in ascending order according to their $u_i$-coordinates. The skeleton of an RBRT ${\cal T}_l$ is independent of the position of the points in $\mathbb{R}^d$ and depends on the \textit{ranks} of the points in each of the $u_i$-coordinates. The rank of a point in a tree at level $i$ of the RBRT ${\cal T}_l$ is its position in the sorted list of all the points ordered by their $u_i$-coordinates. Any tree at any level of the RBRT ${\cal T}_l$ is a balanced binary tree, and no matter how many points are in the tree, it is a tree on $n$ ranks. The following gives the complexity of an RBRT ${\cal T}_l$.

\begin{theorem}\label{the:RBRTcomplexity}{\tt \cite{Abam:2011:KSX:1971362.1971367}}
An RBRT ${\cal T}_l$ uses $O(n\log^dn)$ storage and can be constructed in $O(n\log^d n)$ time. It can be described as a set of pairs $\Psi_l=\{(B_1,R_1),...,(B_m,R_m)\}$ with the following properties. 
\begin{itemize}
\item Each pair $(B_j,R_j)\in \Psi_l$ is generated from an internal node or a leaf node of a tree at level $d$ of ${\cal T}_l$.
\item For any two points $p$ and $q$ in $P$ where $q\in C_l(p)$, there is a unique pair $(B_j,R_j)\in \Psi_l$ such that $p\in B_j$ and $q\in R_j$.
\item For any pair $(B_j,R_j)\in \Psi$, if $p\in B_j$ and $q\in R_j$, then $q\in C_l(p)$ and $p\in \bar{C}_l(q)$. Here, $\bar{C}_l(q)$ is the reflection of $C_l(q)$ through $p$, which is intuitively formed by following the lines through $p$ in the half-spaces of $C_l(q)$.
\item Each point $p\in P$ is in $O(\log^d n)$ pairs of $(B_j,R_j)$, which implies that the number of elements of all the pairs $(R_j,B_j)$ is $O(n\log^d n)$.
\item For any point $p\in P$, all the sets $B_j$ (resp. $R_j$), where $p\in B_j$ (resp. $p\in R_j$), can be found in time $O(\log^d n)$.
\item The set $P\cap C_l(p)$ (resp. $P\cap \bar{C}_l(p)$) of points is the union of $O(\log^d n)$ sets $R_j$ (resp. $B_j$), where the subscript $j$ is such that $p\in B_j$ (resp. $p\in R_j$). 
\end{itemize}
For a set of $n$ moving points, where the trajectories are given by polynomials of degree bounded by a constant, the RBRT ${\cal T}_l$ can be maintained by processing $O(n^2)$ events, each in worst-case time $O(\log^d n)$.
\end{theorem}

\paragraph{Complexity of the $k$-level.} Consider a set of $n$ moving points, where the $y$-coordinate $y_i(t)$ of each point $p_i$ is a polynomial function of at most constant degree $s$. The \textit{$k$-level} of these polynomial functions is a set of points $q\in \mathbb{R}^2$ such that each point $q$ lies on a polynomial function, and such that it is above exactly $k-1$ other polynomial functions; Figure~\ref{fig:k_level} depicts the $3$-level and breakpoints on the $3$-level of four polynomials. The $k$-level tracks the $k^{th}$ lowest point with respect to $y$-axis.

Theorem~\ref{the:totallyDFcomplexity} gives the complexity of the $1$-level (\ie, the number of breakpoints on the lower envelope) for a set of polynomial functions.  

\begin{theorem}\label{the:totallyDFcomplexity}{\tt \cite{Pettie:2013:SBD:2493132.2462390,Agarwal:1995:DSG:868483}}
The number of breakpoints on the $1$-level of $n$ totally-defined (resp. partially-defined), continuous, univariate  functions, such that each pair of them intersects at most $s$ times, is at most $\lambda_s(n)$ (resp. $\lambda_{s+2}(n)$). The sharp bounds on $\lambda_s(n)$ are as follows:

\[
\lambda_s(n) =n\beta(n) =
\begin{cases}
        {n},  & \text{for $s=1$};\\
        {2n-1},  & \text{for $s=2$};\\
        {2n\alpha(n)+O(n)},  & \text{for $s=3$};\\
        {\Theta(n2^{\alpha(n)})},  & \text{for $s=4$};\\
        {\Theta(n\alpha(n)2^{\alpha(n)})},  & \text{for $s=5$};\\
        {n2^{(1+o(1))\alpha^t(n)/t!}}, & \text{for $s\geq 6$};
\end{cases}
\]
here $t={\lfloor {(s-2)/2}\rfloor}$ and  $\alpha(n)$ denotes the inverse Ackermann function.
\end{theorem}

\begin{figure}[h]
\centering
\includegraphics[scale=1]{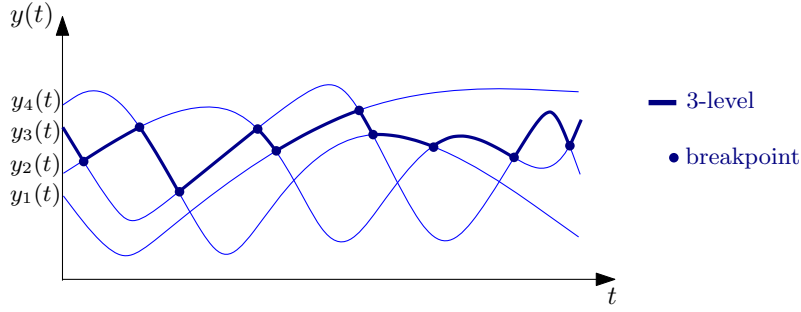}
\caption{The $3$-level of a set of four moving points.}
\label{fig:k_level}
\end{figure}

The following theorem gives the current bounds on the complexity of the $k$-level.

\begin{theorem}\label{the:k_levelComplexity}{\tt \cite{AACS1998,chanii2005,Chan:2008:LAC:1377676.1377691,SharirM1991}}
The complexity of the $k$-level of a set of $n$ partially-defined polynomial functions, such that each pair of them intersects at most $s$ times, is as follows:
\[
\phi(n) = 
\begin{cases}
        {O(n^{3/2} \log n)},  & \text{for $s=2$};\\
        {O(n^{5/3} \text{poly}\log n)},  & \text{for $s=3$};\\
        {O(n^{31/18} \text{poly}\log n)},  & \text{for $s=4$};\\
        {O(n^{161/90-\delta})},  & \text{for $s=5$, for some constant $\delta>0$};\\
        {O(n^{2-1/2s-\delta_s})},  & \text{for odd $s$, for some constant $\delta_s>0$};\\
        {O(n^{2-1/2(s-1)-\delta_s})}, & \text{for even $s$, for some constant $\delta_s>0$}.
\end{cases}
\]
A bound $f(n)$ of $\phi(n)$ can be converted to the $k$-sensitive bound $O(f(k)(n/k)\beta(n/k))$. The complexity of the $(\leq k)$-level is $O(kn\beta(n/k))$. 
\end{theorem}
\paragraph{Maintaining the $k^{th}$ lowest point.} Assume we want to maintain the \textit{$k^{th}$ lowest point} with respect to the $y$-axis among a set $P$ of moving points, where insertions and deletions into the point set $P$ are allowed; the $y$-coordinates of newly inserted points are polynomials of degrees bounded by some constant $s$. 

Using a (dynamic and) \textit{kinetic tournament tree}, one can easily maintain the lowest point. The following summarizes the complexity of this data structure.

\begin{theorem}\label{the:KineticTT} {\tt (Theorem 3.1. of \cite{Agarwal:2008:KDD:1435375.1435379})}
Assume one is given a sequence of $m$ insertions and deletions into a kinetic tournament tree whose maximum size at any time is $n$ (assuming $m\geq n$). The tournament tree generates $O(m\beta_{s+2}(n)\log n)$ events for a total cost of $O(m\beta_{s+2}(n)\log^2 n)$. Each point participates in $O(\log n)$ certificates, so each update/event can be handled in time $O(\log^2 n)$. A kinetic tournament tree on $n$ elements can be constructed in $O(n)$ time.  
\end{theorem}

To maintain the $k^{th}$ lowest point (for any $k\geq 1$) over time, we need to track the order of the moving points, so we use a (dynamic and) \textit{kinetic sorted list}. Each newly inserted point into a kinetic sorted list can exchange its order with other points at most $O(n)$ times.  Thus it is easy to obtain the following.

\begin{theorem}\label{the:KineticSL}
Given a sequence of $m$ insertions and deletions into a kinetic sorted list whose maximum size at any time is $n$. The kinetic sorted list generates $O(mn)$ events. Each point participates in $O(1)$ certificates, so each update/event can be handled in time $O(\log n)$. A kinetic sorted list on $n$ elements can be constructed in $O(n\log n)$ time.  
\end{theorem}

\section{Key Lemmas: Relationships}\label{sec:Relationship}
Here we provide a key insight to obtain the relationships between the proximity problems that are stated in Section~\ref{sec:ProbStatment}.

Consider a polyhedral cone $C_l\in C$ with respect to $\overrightarrow{v}$, where  $C=\{C_0,...,C_{c-1}\}$ is a set of polyhedral cones of opening angle $\theta$ with their apex at the origin $o$ that together cover $\mathbb{R}^d$ (see Lemma~\ref{the:NumofPolyCones}). From now on, we assume that $\theta\leq \pi/3$. Denote by $x_l$ the cone axis of $C_l$ (\ie, the vector in the direction of the unit vector $\overrightarrow{v}$ of $C_l$, $0\leq l\leq c-1$; see Section~\ref{sec:Preliminaries}). Recall that $C_l(p)$ denote a translated copy of $C_l$ with apex at $p$. Denote by $L(P\cap C_l(p))$ the list of the points in $P\cap C_l(p)$, sorted by increasing order of their $x_l$-coordinates.

\begin{lemma}\label{the:keyLemma2}
Let $p_i$ be the $i^{th}$ nearest neighbor of $p$ among a set $P$ of points in $\mathbb{R}^d$, and let $C_{l}(p_i)$ be the cone of $p_i$ that contains $p$. Then point $p$ is among the first $i$ points in $L(P\cap C_l(p_i))$.
\end{lemma}
\begin{proof}
Let $P'=P\backslash \{p_1,...,p_{i-1}\}$. Then point $p_i$ is the closest point to $p$ among the points in $P'$; see Figure~\ref{fig:ProofKeyLemma}(a). It can be proved by contradiction that point $p$ has the minimum $x_l$-coordinate among the points in $P'\cap C_l(p_i)$ (Lemma 8.1 of~\cite{Agarwal:2008:KDD:1435375.1435379}): Assume there is a point $r\in P$ inside the cone $C_l(p_i)$ whose $x_l$-coordinate is less than the $x_l$-coordinate of $p$; see Figure~\ref{fig:ProofKeyLemma}(b) for an example where $i=3$. Consider the triangle $pp_ir$. Since $p_i$ is the closest point to $p$ among the points in $P'$, $|pp_i|<|pr|$, which implies that angle $\angle pp_ir>\angle prp_i$. This is a contradiction, because $\angle pp_ir\leq\pi/3$ and $\angle prp_i>\pi/3$.

Now we add the points $p_1,...,p_{i-2}$, and $p_{i-1}$  to the point set $P'$. Consider the worst case scenario that all these  $i-1$ points insert inside the cone $C_l(p_i)$, and that the $x_l$-coordinates of all these points are less than the $x_l$-coordinate of $p$. Then the point $p$ is still among the  first $i$ points in the sorted list  $L(P\cap C_l(p_i))$.
\end{proof}

\begin{figure}[h]
\centering
\includegraphics[scale=1]{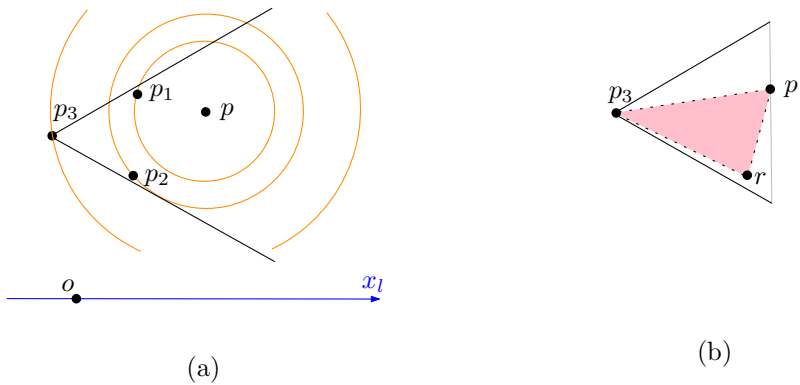}
\caption{Point $p_3$ is the $3^{rd}$ nearest neighbor of $p$. After deleting the points $p_1$ and $p_2$, point $p_3$ is the closest point to $p$; among the points in $C_0(p_3)$, $p$ has the minimum length projection on the bisector $x_0$.}
\label{fig:ProofKeyLemma}
\end{figure}

Consider the \textit{$k$-nearest neighbor graph} (\knng) of a point set $P$, which is constructed by connecting each point in $P$ to all its $k$-nearest neighbors. Let ${\cal K}_l(p)$ be the set of the first $k$ points in the sorted list $L(P\cap C_l(p))$. If we connect each point $p\in P$ to the points in ${\cal K}_l(p)$, for $l=0,...,c-1$, we obtain what we call the \textit{$k$-Semi-Yao graph}\footnote{Rahmati~\etal~\cite{Rahmati2014} called the theta graph the \textit{Semi-Yao graph} because of its close relationship to the Yao graph~\cite{DBLP:journals/siamcomp/Yao82}. Here, we call the generalization of the Semi-Yao graph, with respect to $k$, the \textit{$k$-Semi-Yao graph}.} (\ksyg). The \ksyg~has the following property. 

\begin{lemma}\label{the:NNGsubSYG}
The \knng~of a point set $P$ in $\mathbb{R}^d$ is a subgraph of the \ksyg~of $P$.
\end{lemma}
\begin{proof}
Lemma~\ref{the:keyLemma2} gives a necessary condition for $p_k$ to be the $k^{th}$ nearest neighbor of $p$: $p\in {\cal K}_l(p_k)$, where $l$ is such that $p\in C_l(p_k)$.  Therefore, the edge set of the \ksyg~covers the edges of the \knng.
\end{proof}

Now we obtain the following, for answering \rknn~queries.

\begin{lemma}\label{the:RkNNsUprBnd}
The set of reverse $k$-nearest neighbors of a query point $q\notin P$ is a subset of the union of the sets ${\cal K}_l(q)$, for $l=0,...,c-1$. The number of reverse $k$-nearest neighbors of the query point $q$ is upper-bounded by $O(k)$.
\end{lemma}
\begin{proof}
Assume, among the points in $P\cup\{q\}$, that $q$ is the $i^{th}$ nearest neighbor of some point $p$, where $i\leq k$. There exists a cone $C_l(q)$ of $q$ such that $p\in C_l(q)$. From Lemma~\ref{the:keyLemma2}, $p\in {\cal K}_l(q)$. Therefore, each of the $k$-reverse nearest neighbors of $q$ is in the union of ${\cal K}_l(q)$, $l=0,...,c-1$. 

We assume $d$ is arbitrary but fixed, so $c$ is a constant. Thus the cardinality of the union of ${\cal K}_l(q)$ is $O(k)$, which implies that the number of reverse $k$-nearest neighbors is upper-bounded by $O(k)$. 
\end{proof}
\section{Computing the \ksyg~and All $k$-Nearest Neighbors}\label{sec:ReportKNNs}
Here we first describe how to compute the \ksyg, which will aid in understanding how our kinetic approach works. Then, via a construction of the \ksyg, we give a simple method for reporting all the $k$-nearest neighbors. 

To efficiently construct the \ksyg, we need a data structure to perform the following operation efficiently: For each $p\in P$ and any of its cones $C_l(p)$, $0\leq l\leq c-1$, find ${\cal K}_l(p)$, the set of the first $k$ points in the sorted list $L(P\cap C_l(p))$. Such an operation can be performed by using range tree data structures. For each cone $C_l$, we construct an associated $d$-dimensional range tree ${\cal T}_l$ as follows.

Consider a particular cone $C_l$ with apex at $o$; see Figure~\ref{fig:aConeC_l}(a). The cone $C_l$ is the intersection of $d$ half-spaces $f^+_1,...,f^+_d$  with coordinate axes $u_1,...,u_d$.

The range tree ${\cal T}_l$ is a regular $d$-dimensional range tree based on the $u_i$-coordinates (see~\cite{Berg:2008:CGA:1370949}). The points at level $i$ are sorted at the leaves according to their $u_i$-coordinates. Any $d$-dimensional range tree, \eg, ${\cal T}_l$, uses $O(n\log^{d-1} n)$ space and can be constructed in time $O(n\log^{d-1} n)$, and for any point $r\in \mathbb{R}^d$, the points of $P$ inside the query cone $C_l(r)$ whose sides are parallel to $f_i$, $1\leq i\leq d$, can be reported in time $O(\log^{d-1} n + z)$, where $z$ is the cardinality of the set $P\cap C_l(r)$~\footnote{For a set of stationary points, there are lots of improvements for answering rectangular range queries (\eg, see~\cite{Chan:2011:ORS:1998196.1998198}).}.

Now we add a new level to ${\cal T}_l$, based on the coordinate $x_l$. To find ${\cal K}_l(p)$ in an efficient time, we use the level $(d+1)$ of ${\cal T}_l$, which is constructed as follows: For each internal node $v$ at level $d$ of ${\cal T}_l$, we create a list $L(R(v))$ sorted by increasing order of $x_l$-coordinates of the points in $R(v)$. For the set $P$ of $n$ points in $\mathbb{R}^d$, the modified range tree ${\cal T}_l$, which now is a $(d+1)$-dimensional range tree, uses $O(n\log^d n)$ space and can be constructed in time $O(n\log^d n)$~\cite{Berg:2008:CGA:1370949}.
 
The following establishes the processing time for obtaining a set ${\cal K}_l(p)$. 

\begin{lemma}\label{the:SortingLists}
Given ${\cal T}_l$, the set ${\cal K}_l(p)$ can be found in time $O(\log^d n+k)$.
\end{lemma}
\begin{proof}
The set $P\cap C_l(p)$ is the union of $\hat{m}=O(\log^d n)$ sets $R(v)$, where $v$ ranges over internal nodes $v$ at level $d$ of ${\cal T}_l$. Consider the associated sorted lists $L(R(v))$. Given $\hat{m}$ sorted lists $L(R(v))$, the $k^{th}$ point in $L(P\cap C_l(p))$ can be obtained in time $O(\hat{m}+k)$ (Theorem 1 of~\cite{Frederickson1982197}). 

By examining the points in each of the $\hat{m}$ sorted lists whose $x_l$-coordinates are less than or equal to the $x_l$-coordinate of the $k^{th}$ point, we can find the members of ${\cal K}_l(p)$  in time $O(k)$.
\end{proof}

By Lemma~\ref{the:SortingLists}, we can find all the ${\cal K}_l(p)$, for all $p\in P$. This gives the following.

\begin{corollary}\label{the:kSYG_Construction}
Using a data structure of size $O(n\log^d n)$, the edges of the $\ksyg$ of a set of $n$ points in $\mathbb{R}^d$ can be reported in time $O(n\log^d n+kn)$.
\end{corollary}

Now we state and prove the cost of reporting all the $k$-nearest neighbors in our approach, which in fact derives the known results in a new way~\footnote{For $k=\Omega(\log^{d-1}n)$, both our data structure and the best previous data structure~\cite{Dickerson:1996:APP:236464.236474} have the same complexity for reporting all the $k$-nearest neighbors. Arya~\etal~\cite{Arya:1998:OAA:293347.293348} have a kd-tree implementation to approximate the nearest neighbors of a query point that is in use by practitioners~\cite{10.1109/TVCG.2010.9} who have found it challenging to implement the theoretical algorithms~\cite{Vaidya:1989:ONL:70530.70532,Callahan366854,Clarkson:1983:FAN:1382437.1382825,Dickerson:1996:APP:236464.236474}. Since to report all the $k$-nearest neighbors ordered by distance from each point our method uses multidimensional range trees, which can be easily implemented, we believe our method may be useful in practice.}.


\begin{theorem}\label{the:kNNG_Construction}
For a set of $n$ points in $\mathbb{R}^d$, our data structure can report all the $k$-nearest neighbors, in order of increasing distance from each point, in time $O(n\log^d n + kn\log n)$. The data structure uses $O(n\log^d n + kn)$ space.
\end{theorem}
\begin{proof}
Suppose we are given the \ksyg~(see Corollary~\ref{the:kSYG_Construction}), which is a supergraph of the \knng~(from Lemma~\ref{the:NNGsubSYG}), and we want to report all the $k$-nearest neighbors. 

Let $E_p$ be the set of edges incident to the point $p$ in the \ksyg. By sorting these edges in non-decreasing order according to their Euclidean lengths, which can be done in time $O(|E_p|\log |E_p|)$, we can find the $k$-nearest neighbors of $p$  ordered by increasing Euclidean distance from $p$. 

Since the number of edges in the $\ksyg$ is $O(kn)$ and each edge $pq$ belongs to exactly two sets $E_p$ and $E_q$, the time to find all the $k$-nearest neighbors, for all the points $p\in P$, is $\sum_{p} O(|E_p|\log |E_p|) = O(kn\log n)$. The proof obtains by combining this with the results of Corollary~\ref{the:kSYG_Construction}.
\end{proof}

\section{Kinetic $k$-Semi-Yao Graph}\label{sec:kineticKSYG}
In Section~\ref{sec:KDSfor1SYG}, we first provide a KDS for the \ksyg, for $k=1$. Then in Section~\ref{sec:KDSforkSYG} we extend our kinetic approach to any $k\geq 1$.

\subsection{The case $k=1$}\label{sec:KDSfor1SYG}

The \1syg~remains unchanged as long as the order of the points in each of the coordinates $u_1,...,u_d$, and $x_l$ associated to each cone $C_l\in C$ remains unchanged. Therefore, to track the changes to the \1syg~over time, we distinguish between two types of events:

\begin{itemize}
\item \textbf{$u$-swap event:} Such an event occurs if two points exchange their order in the $u_i$-coordinate. 
\item \textbf{$x$-swap event:} This event occurs whenever two points exchange their order in the $x_l$-coordinate. 
\end{itemize}

The $u$-swap events can be tracked by defining $d$ kinetic sorted lists $L(u_1),...,L(u_d)$ of the points for each of the coordinates $u_1,...,u_d$ (see Section~\ref{sec:Preliminaries}). In addition, to track the $x$-swap events, we create a kinetic sorted list $L(x_l)$ of the points with respect to the $x_l$-coordinates of the points.

Fix a cone $C_l\in C$, $0\leq l\leq c-1$. Corresponding to the cone $C_l$, we create kinetic ranked-based range trees (RBRTs) ${\cal T}_l$ (see Section~\ref{sec:Preliminaries}). Consider the corresponding cone separated pair decomposition (CSPD) $\Psi_l=\{(B_1,R_1),...,(B_m,R_m)\}$ of ${\cal T}_l$. Let $r_j$ be the point with minimum $x_l$-coordinate among the points in $R_j$. Denote by $\ddot{w}_l$ the point in $P\cap C_l(w)$ with minimum $x_l$-coordinate; in fact $\ddot{w}_l$ is the point with the minimum $x_l$-coordinate among the points $r_j$, where the subscripts $j$ are such that $P\cap C_l(w)=\bigcup_j R_j$. Note that to maintain the \1syg, for each point $w\in P$, in fact we must track $\ddot{w}_l$.
To apply required changes to $\ddot{w}_l$ for all $w\in P$, when an event occurs, in addition $r_j$, we need to maintain more information for each subscript $j$ (\ie, at each internal node $v$ at level $d$ of ${\cal T}_l$). The next paragraph describes the extra information.

Allocate a \textit{label} to each point in $P$. Let $L(B_j)$ be a sorted list of the points $w\in B_j$ according to the labels of their $\ddot{w}_l$. This sorted list is used to answer the following query while processing $x$-swap events: Given a query point $p$, find all the points $w\in B_j$ such that $\ddot{w}_l=p$. Since we perform  updates (insertions/deletions) to the sorted lists $L(B_j)$ over time, we implement them using a dynamic binary search tree (\eg, a \textit{red-black tree}); each update is performed in worst-case time $O(\log n)$. Furthermore, for each $w\in P$, we create a set of links to $w$ in the sorted lists $L(B_j)$; denote this set by $Link(w)$; we use this set to efficiently delete a point $w$ from the sorted lists $L(B_j)$ when we are handling the events.

In the preprocessing step before the motion, for any subscript $j$ and for any point $w\in P$, we find $r_j$ and $\ddot{w}_l$, and then we construct $L(B_j)$ and $Link(w)$.

\begin{lemma}\label{the:PreProcStep}
Our KDS uses  $O(n\log^dn)$ space and $O(n\log^{d+1}n)$ preprocessing time.
\end{lemma}
\begin{proof}
By Theorem~\ref{the:RBRTcomplexity}, each point $p\in P$ is in at most $O(\log^d n)$ sets $B_j$, and $O(\log^dn)$ sets $R_j$, so the cardinality of each set $Link(p)$ is $O(\log^d n)$, and the size of sets $B_j$ and $R_j$, for all $j$, is $O(n\log^dn)$. This implies that $(i)$ the KDS uses $O(n\log^dn)$ storage, $(ii)$ we can find all the $r_j$ and $\ddot{w}_l$ in time $O(n\log^d n)$, and $(iii)$  we can sort the points $w$ in all the $B_j$ according to the labels of their $\ddot{w}_l$ in $O(n\log^{d+1}n)$ time, and then by tracing the members of the sorted lists $L(B_j)$, we can create $Link(p)$ for all $p\in P$ in the same time  $O(n\log^{d}n)$.
\end{proof}

Now let the points move. The following shows how to maintain and reorganize $Link(w)$, $L(B_j)$ and $r_j$, for any subscript $j$ and for any point $w\in P$,  when a $u$-swap event or an $x$-swap event occurs. Note that maintenance of the sets $Link(w)$, for all $w\in P$, in fact gives a kinetic maintenance of the \1syg.

\begin{figure}[t]
\centering
\includegraphics[scale=1]{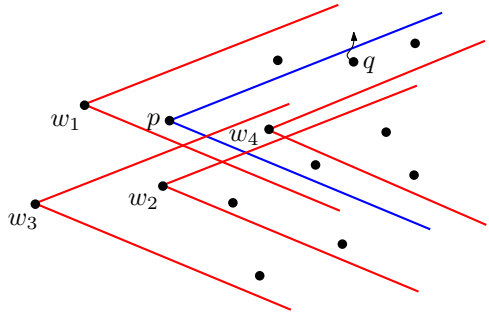}
\caption{A $u$-swap between $p$ and $q$ does not change the memberships of points in other cones.}
\label{fig:Uswap}
\end{figure}

\paragraph{Handling $u$-swap events.} Consider a $u$-swap between $p$ and $q$. Without loss of generality, assume $q\in C_l(p)$ before the event; see Figure~\ref{fig:Uswap}. After the event, $q$ moves outside the cone $C_l(p)$. Note that this event does not change the points in $P\cap C_l(w)$ for other points $w\in P$. Therefore, the only change that might happen to the \1syg~is to replace an edge incident to $p$ inside the cone $C_l(p)$ with a new one. In particular, when two points $p$ and $q$ exchange their order with respect to the $u_i$-coordinate, we perform the following steps.
\begin{itemize}
\item[\texttt{U1)}] We update the kinetic sorted list $L(u_i)$.
\item[\texttt{U2)}] A $u$-swap event may change the structure of the RBRT ${\cal T}_l$, so we update ${\cal T}_l$.
\item[\texttt{U3)}] We delete the point(s) $p$ from the sorted lists $L(B_j)$ where $p\in B_j$.
\item[\texttt{U4)}] We delete the members of $Link(p)$.
\item[\texttt{U5)}] We update the values in $\{r_j~|~p\in R_j~\vee~q\in R_j\}$.
\item[\texttt{U6)}] We find the point $\ddot{p}_l$ in $P\cap C_l(p)$ whose $x_l$-coordinate is minimum among all the $r_j$ such that $p\in B_j$.
\item[\texttt{U7)}] We add the point $p$ into all the sorted lists $L(B_j)$ according to the label of the new value of $\ddot{p}_l$. Then we construct the set $Link(p)$, which in fact is the new set of links to $p$ in the sorted lists $L(B_j)$.
\end{itemize}

The following lemma gives the complexity of the steps \texttt{U1},...,\texttt{U7} above.

\begin{lemma}\label{the:Uswap}
For maintenance of the \1syg, our KDS handles $O(n^2)$ $u$-swap events, each in worst-case time $O(\log^{d+1} n)$.
\end{lemma}
\begin{proof}
For a fixed dimension $d$, (by Theorem~\ref{the:KineticSL}) the kinetic sorted lists $L(u_i)$, $1\leq i \leq d$, handle $O(n^2)$ events, each in $O(\log n)$ time (Step \texttt{U1}). 

From Theorem~\ref{the:RBRTcomplexity}, an update to ${\cal T}_l$ takes $O(\log^d n)$ time (Step \texttt{U2}). By using the links in $Link(p)$, Step \texttt{U3} can be done in $O(\log^{d+1} n)$ time. 

By Theorem~\ref{the:RBRTcomplexity}, all the $R_j$ can be found in $O(\log^d n)$ time, so the values $r_j$ can be updated in $O(\log^d n)$ worst-case time  (Step \texttt{U5}); also, since each point is in $O(\log^dn)$ sets $B_j$, Step \texttt{U6} takes $O(\log^d n)$ time. 

Each operation in a sorted list $L(B_j)$ can be done in $O(\log n)$ time; this implies that Step \texttt{U7} takes $O(\log^{d+1} n)$ time.
\end{proof}

\paragraph{Handling $x$-swap events.} Denote by $x_l(p)$ the $x_l$-coordinate of $p$. Let $p$ and $q$ be two consecutive points with $p$ preceding $q$ (\ie, $x_l(p)<x_l(q)$) before the $x$-swap event. The structure of ${\cal T}_l$ remains unchanged when an $x$-swap event between $p$ and $q$ occurs. Such an event might change the value of $\ddot{w}_l$ of some points $w$ of the sorted lists $L(B(.))$ and if so, we must find such points $w$ and apply the required changes. 

The number of all changes to the \1syg~depends on how many points $w\in P$ have both $p$ and $q$ in their cones $C_l(w)$. Note that, while reporting the points in $P\cap C_l(w)$ for $w$, both $p$ and $q$ might be in the same set $R_j$ (see Figure~\ref{fig:Xswap}(a)) or in two different sets $R_j$ and $R_{\bar{j}}$ (see Figure~\ref{fig:Xswap}(b)). To find such points $w$, when an $x$-swap event between $p$ and $q$ occurs, we seek (I) subscripts $j$ where $\{p,q\}\subseteq R_j$, and (II) subscripts $j$ and $\bar{j}$ where $p\in R_j$ and $q\in R_{\bar{j}}$. In the first case, we must find any point $w\in B_j$ such that $\ddot{w}_l=p$ (\ie, $p$ is the point with minimum $x_l$-coordinate in the cone $C_l(w)$). Then we replace $p$ by $q$ after the event: $\ddot{w}_l=q$. This means that we replace the edge $wp$ of the \1syg~with $wq$.

\begin{figure}[t]
\centering
\includegraphics[scale=1]{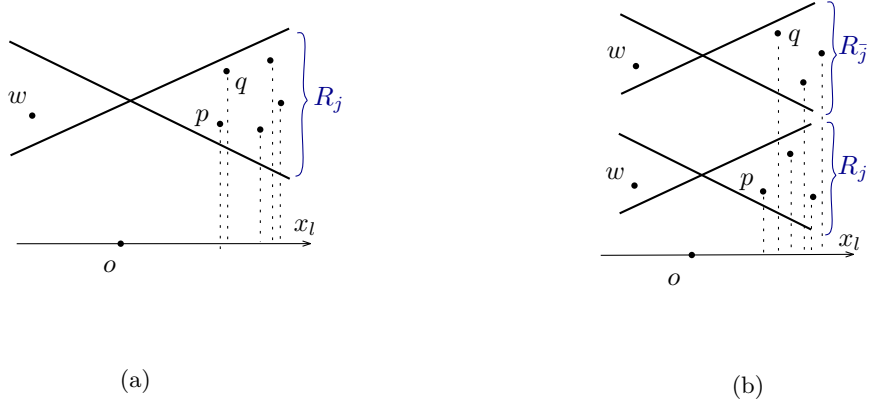}
\caption{Two cases when an $x$-swap between $p$ and $q$ occurs.}
\label{fig:Xswap}
\end{figure}

Note that in the second case there is no point $r\in B_j$ such that $\ddot{r}_l=q$, because $x_l(p)<x_l(q)$. Also note that if there is a point $w\in B_j$ such that $\ddot{w}_l=p$, we change the value of $\ddot{w}_l$ to $q$ if $q\in C_l(w)$; in the case that $q\in C_l(w)$, there is a unique pair $(B_{\bar{j}},R_{\bar{j}})$ where $w\in B_{\bar{j}}$ and $q\in R_{\bar{j}}$. Thus we can find $w$ in the set $B_{\bar{j}}$ and we do not need to check whether point $w$ is in $B_j$ or not. In particular, for the second case, we only need to check whether there is a point $w\in B_{\bar{j}}$ such that $\ddot{w}_l=p$; if so, we change the value of $\ddot{w}_l$ to $q$ ($\ddot{w}_l=q$).

From the above discussion, the following steps summarize the update mechanism of our KDS for maintenance of the \1syg~when an $x$-swap event occurs.

\begin{itemize}
\item[\texttt{X1)}] We update the kinetic sorted list $L(x_l)$.
\item[\texttt{X2)}] We find all the subscripts $j$ such that $\{p,q\}\subseteq R_j$ and $r_j=p$. Also, we find all the subscripts $j$ where $r_j=q$ (see Figure~\ref{fig:Xswap}).
\item[\texttt{X3)}] For each subscript $j$ (from Step \texttt{X2}), we find all the points $w$ in the sorted list $L(B_j)$ where $\ddot{w}_l=p$.
\item[\texttt{X4)}] For each $w$ (from Step \texttt{X3}), using the links in $Link(w)$, we update all the corresponding sorted lists $L(B_j)$: we delete $w$ from them, change the value of $\ddot{w}_l$ to $q$, and add $w$ into the sorted lists according to the label of $q$.
\end{itemize}

The number of edges incident to a point $p$ in the \1syg~is $O(n)$. Thus when an $x$-swap event between $p$ and some point $q$ occurs, it might cause $O(n)$ changes to the \1syg. The following lemma proves that an $x$-swap event can be handled in polylogarithmic amortized time.

\begin{lemma}\label{the:Xswap}
For maintenance of the \1syg, our KDS handles $O(n^2)$ $x$-swap events with total processing time $O(n^2\beta_{s+2}(n)\log^{d+1} n)$.
\end{lemma}
\begin{proof} 
From Theorem~\ref{the:KineticSL}, Step \texttt{X1} takes $O(\log n)$ time.  By Theorem~\ref{the:RBRTcomplexity}, all the subscripts $j$ at Step \texttt{X2} can be found in $O(\log^d n)$ time.

For each $j$ of Step \texttt{x3}, the update mechanism spends $O(\log n + z_j)$ time where $z_j$ is the number of all the points $w\in B_j$ such that $\ddot{w}_l=p$. For all the subscripts $j$, the second step takes $O(\log^{d+1} n + \sum_j z_j)$ time. Note that $\sum_j z_j$ is equal to the number of exact changes to the \1syg. Since the number of exact changes to the ~\1syg~of a set of $n$ moving points in a fixed dimension $d$ is $O(n^2\beta_{s+2}(n))$ (Theorem 6 of~\cite{Rahmati2014}), the total processing time of Step \texttt{X3} for all the $O(n^2)$ $x$-swap events is $O(n^2\log^{d+1} n + n^2\beta_{s+2}(n))=O(n^2\log^{d+1} n)$. 

For each $w$ of Step \texttt{X4}, the processing time to apply changes to the KDS is $O(\log^{d+1} n)$. For each $w$ of \texttt{X4}, it is in fact a change to the \1syg. Thus the update mechanism spends $O(n^2\beta_{s+2}(n)\log^{d+1} n)$ time to handle all the $O(n^2)$ events.

Summing over the complexities of Steps \texttt{X1}-\texttt{X4}, for all the $x$-swap events, gives the total processing time  $O(n^2\beta_{s+2}(n)\log^{d+1} n)$.
\end{proof}

Now we obtain the main result of this section, which summarizes the complexity of the proposed KDS for the \1syg.

\begin{theorem}\label{the:KineticSYG}
Our KDS for maintenance of the \1syg~of a set of $n$ moving points in $\mathbb{R}^d$, where the coordinates of each point are given by polynomials of at most constant degree $s$, uses $O(n\log^d n)$  space, $O(n\log^{d+1} n)$ preprocessing time, and handles $O(n^2)$ events with a total cost of $O(n^2\beta_{s+2}(n)\log^{d+1} n)$. The KDS is compact, efficient, amortized responsive, and local.
\end{theorem}
\begin{proof}
From Lemma~\ref{the:PreProcStep}, the KDS uses  $O(n\log^{d+1} n)$ preprocessing time and $O(n\log^d n)$  space. The total cost to process all the $O(n^2)$ events is $O(n^2\beta_{2s+2}(n)\log^{d+1} n)$ (by Lemmas~\ref{the:Uswap} and ~\ref{the:Xswap}); this implies that the KDS is amortized responsive. 

Since the number of the certificates is $O(n)$, the KDS is compact. 

Each point in a kinetic sorted list participates in two certificates, one created with the previous point and one with the next point, which implies the KDS is local. 

Since the number of the external events is $O(n^2\beta_{s+2}(n))$ and the number of the events that the KDS processes is $O(n^2)$, the KDS is efficient.
\end{proof}

\subsection{The general case: any $k\geq 1$}\label{sec:KDSforkSYG}
Here, we extend our kinetic approach to maintain the \ksyg, for any $k\geq 1$.

For maintenance of the \ksyg~over time, we must track the sets ${\cal K}_l(p)$, $0\leq l\leq c-1$, for each point $p\in P$. In order to do this, for each subscript $j\in \{1,...,m\}$, we need to maintain a list $L(R_j)$ of the points in $R_j$, sorted in ascending order according to their $x_l$-coordinates over time. Note that each set $R_j$ is some $R(v)$, the set of points at the leaves of the subtree rooted at some internal node $v$ at level $d$ of the RBRT ${\cal T}_l$. To maintain these sorted lists $L(R_j)$, we add a new level to the RBRT ${\cal T}_l$; the points at the new level are sorted at the leaves in ascending order according to their $x_l$-coordinates. Therefore, for updating the modified RBRT ${\cal T}_l$, in addition to the $u$-swap events, we handle the $x$-swap events as well. The modified RBRT ${\cal T}_l$ behaves like a $(d+1)$-dimensional RBRT. From Theorem~\ref{the:RBRTcomplexity}, when a $u$-swap event or an $x$-swap event occurs, ${\cal T}_l$ can be updated in worst-case time $O(\log^{d+1}n)$.

Denote by $\ddot{p}_{l,k}$ the $k^{th}$ point in $L(P\cap C_l(p))$.  To track and update the points in ${\cal K}_l(p)$, for all the points $p\in P$, we maintain the following over time:

\begin{itemize}
\item A set of $d+1$ kinetic sorted lists $L(u_i)$, $i=1,...,d$, and $L(x_l)$ of the points in $P$. We use these sorted lists to track the order of the points in the coordinates $u_i$, $1\leq i\leq d$, and $x_l$, respectively.
\item For each $B_j$, a sorted list $L(B_j)$ of the points in $B_j$. The order of the points $p$ in $L(B_j)$ is according to the labels of their $\ddot{p}_{l,k}$. This sorted list $L(B_j)$ is used to efficiently answer the following query: Given a query point $q$ and a $B_j$, find all $p\in B_j$ such that $\ddot{p}_{l,k}=q$.
\item The $k^{th}$ point $r_{j,k}$ in the sorted list $L(R_j)$. We maintain the values $r_{j,k}$ in order to make necessary changes to the \ksyg~when an $x$-swap event occurs.
\end{itemize}

As in Section~\ref{sec:KDSfor1SYG}, when the points move, we handle two types of events, \textit{$u$-swap events} and \textit{$x$-swap events}.

\paragraph{Handling $u$-swap events.}
Let $q\in C_l(p)$ before the $u$-swap event. Whenever the two points $p$ and $q$ exchange their $u_i$-order, the only change that might occur is the replacement of a member of ${\cal K}_l(p)$ with a new one. In particular, when such an event occurs, we perform the following updates.

\begin{enumerate}
\item[\texttt{U1)}] We update the kinetic sorted list $L(u_i)$.
\item[\texttt{U2)}] We update ${\cal T}_l$ . If a point is deleted or inserted into a $B_j$, we update the corresponding sorted list $L(B_j)$.
\item[\texttt{U3)}] After updating ${\cal T}_l$, a point $q$ might be inserted or deleted from some $R_j$ and change the values of $r_{i,k}$. For all $R_j$ where $q\in R_j$, before and after the event, we perform the following. We check whether the $x_l$-coordinate of  $q$ is less than or equal to the $x_l$-coordinate of $r_{j,k}$; if so, we take the successor or predecessor point of $r_{j,k}$ in $L(R_j)$ as the new value for $r_{j,k}$.
\item[\texttt{U4)}] We query to find ${\cal K}_l(p)$.
\item[\texttt{U5)}] If we obtain a new value for $\ddot{p}_{l,k}$, which in fact is the point with maximum $x_l$-coordinate among the points in ${\cal K}_l(p)$, we update all $L(B_j)$ such that $p\in B_j$.
\end{enumerate}

Now the following gives the complexity of handling $u$-swap events.

\begin{lemma}\label{the:UswapEvents}
Our KDS for maintenance of the \ksyg~handles $O(n^2)$ $u$-swap events, each in worst-case time $O(\log^{d+1}n +k)$.
\end{lemma}
\begin{proof}
Each swap event in a kinetic sorted list can be handled in $O(\log n)$ time (Step \texttt{U1}). Since each update (insertion/deletion) to $L(B_j)$ takes $O(\log n)$ time, and since each point is in $O(\log^d n)$ sets $B_j$, Step \texttt{U2} takes $O(\log^{d+1}n)$ time. It is obvious that the processing time of Steps \texttt{U3} and \texttt{U5} is $O(\log^{d+1} n)$. From Lemma~\ref{the:SortingLists}, Step \texttt{U4} takes $O(\log^d n +k)$ time.

The trajectories of the points are given by  bounded degree polynomials, so the number of events, \ie, changes to the order of the points, is $O(n^2)$.
\end{proof}

\paragraph{Handling $x$-swap events.}
Consider an $x$-swap event between two consecutive points $p$ and $q$ with $p$ preceding $q$. This event does not change the elements of the pairs $(B_j,R_j)\in \Psi_l$, but this event changes the $k$-SYG if both $p$ and $q$ are in the cone $C_l(w)$, for some $w$ such that $\ddot{w}_{l,k}=p$. In particular, we perform the following updates when two points $p$ and $q$ exchange their $x_l$-order.
\begin{enumerate}
\item[\texttt{X1)}] We update the kinetic sorted list $L(x_l)$; this takes $O(\log n)$ time.
\item[\texttt{X2)}] We update ${\cal T}_l$, which takes $O(\log ^{d+1} n)$ time.
\item[\texttt{X3)}] We find all the sets $R_j$ where both $p$ and $q$ belong to $R_j$ and such that $r_{j,k}=p$. Also, we find all the sets $R_j$ where $r_{j,k}=q$. This step takes $O(\log^d n)$ time.
\item[\texttt{X4)}] For each $R_j$ (from Step \texttt{X3}), we extract all the points $w$ in the sorted lists $L(B_j)$ such that $\ddot{w}_{l,k}=p$. Note that each change to the value of $\ddot{w}_{l,k}$ is a change to the \ksyg.
\item[\texttt{X5)}] For each $w$ (from Step \texttt{X4}), we update all the sorted lists $L(B_j)$ where $w\in B_j$: we delete $w$ from the sorted lists $L(B_j)$, update the previous value of $\ddot{w}_{l,k}$, which is $p$, by the new value $q$, and add $w$ back to the sorted lists $L(B_j)$ according to the label of $\ddot{w}_{l,k}$.
\end{enumerate}

To prepare for Lemma~\ref{the:XswapEvents} below, which summarizes the complexity of handling $x$-swap events, we first give, in Lemma~\ref{the:allKSYGchanges}, an upper bound for the number of changes to the \ksyg~of a set of moving points.

\begin{lemma}\label{the:allKSYGchanges}
The number of changes to the \ksyg~of a set of $n$ moving points, where the coordinates of each point are given by polynomial functions of at most constant degree $s$, is $\chi_k=O(n\phi(n))$, where $\phi(n)$ denotes the complexity of the $k$-level of  partially-defined polynomial functions of degree bounded by some constant $s$. 
\end{lemma}
\begin{proof}
Fix a point $p\in P$ and one of its cones $C_l(p)$. There are $O(n)$ insertions/deletions into the cone $C_l(p)$ over time. The $x_l$-coordinates of these points create $O(n)$ partial functions. The \ksyg~changes if a change to $\ddot{p}_{l,k}$ occurs. The number of all changes to $\ddot{p}_{l,k}$ is equal to  the complexity of the $k$-level of these $O(n)$ partial functions.

Therefore, summing over all the $n=|P|$ points, the number of changes to the \ksyg~is within a linear factor of $\phi(n)$: $\chi_k=O(n\phi(n))$. 
\end{proof}

\begin{lemma}\label{the:XswapEvents}
Our KDS for maintenance of the \ksyg~handles $O(n^2)$ $x$-swap events with a total cost of $O(n\phi(n)\log^{d+1}n)$.
\end{lemma}
\begin{proof}
The complexities of the first three steps are clear. For each found $R_j$ from Step \texttt{X3}, Step \texttt{X4} takes $O(\log n + \xi_j)$ time, where $\xi_j$ is the number of points $w\in B_j$ such that $\ddot{w}_{l,k}=p$. Thus, for all the $O(\log^d n)$ sets $R_j$ of Step \texttt{X3}, Step \texttt{X4} takes $O(\log^{d+1}n +\sum_j \xi_j)$ time, where $\sum_j \xi_j$ is the number of exact changes to the \ksyg. Therefore, for all the $O(n^2)$ $x$-swap events, the total processing time for this step is $O(n^2\log^{d+1}n+\chi_k)=O(\chi_k)$.

The processing time for Step \texttt{X5} is a function of $\chi_k$. For each change to the \ksyg, this step spends $O(\log^{d+1}n)$ time to update the sorted lists $L(B_j)$. Thus the total processing time for all the $x$-swap events in this step is $O(\chi_k*\log^{d+1}n)$.
\end{proof}

Now we can obtain the following.
\begin{theorem}\label{the:KinetickSYG}
For a set of $n$ moving points in $\mathbb{R}^d$, where the coordinates of each point are polynomial functions of at most constant degree $s$, our \ksyg~KDS uses $O(n\log^{d+1} n)$ space and $O(n\log^{d+1} n)$ preprocessing time, and handles $O(n^2)$ events with a total cost of $O(kn^2 + n\phi(n)\log^{d+1}n)$.
\end{theorem}
\begin{proof}
The proof obtains by combining the results of Theorem~\ref{the:RBRTcomplexity} and Lemmas~\ref{the:UswapEvents} and~\ref{the:XswapEvents}.
\end{proof}
\section{The Applications}\label{sec:applications}
\subsection{Kinetic All $k$-Nearest Neighbors}\label{sec:app_kNNs}
Let us be given a KDS for the \ksyg, a supergraph of the \knng~(from Theorems~\ref{the:KinetickSYG} and~\ref{the:KineticSYG}). This section shows how to maintain all the $k$-nearest neighbors over time. We first consider the case $k=1$, and then the general case, for any $k\geq 1$.

\paragraph{The case $k=1$.}\label{sec:KDSfor1NNs}
We use dynamic and kinetic tournament trees (see Section~\ref{sec:Preliminaries}) to maintain all the $1$-nearest neighbors. For each point $p$ in the \1syg, we create a dynamic and kinetic tournament tree $TT_p$, whose elements are the edges incident to $p$ in the \1syg. 

The following gives the complexity of our KDS for all $1$-nearest neighbors.

\begin{theorem}\label{the:KineticAllNN}
Our KDS for maintenance of all the $1$-nearest neighbors of a set of $n$ moving points in $\mathbb{R}^d$, where the coordinates of each point are polynomial functions of at most constant degree $s$, has the following properties. 
\begin{enumerate}
\item The KDS uses $O(n\log^d n)$ space and $O(n\log^{d+1} n)$ preprocessing time.
\item It processes $O(n^2)$ $u$-swap events, each in worst-case time $O(\log^{d+1} n)$.
\item It processes $O(n^2)$ $x$-swap events, for a total cost of $O(n^2\beta_{2s+2}(n)\log^{d+1} n)$.
\item The KDS processes $O(n^2\beta_{2s+2}^2(n)\log n)$ tournament events, and processing all the events takes $O(n^2\beta_{2s+2}^2(n)\log^2 n)$ time.
\item The KDS is efficient, amortized responsive, compact, and each point participates in $O(1)$ certificates on average.
\end{enumerate}
\end{theorem}
\begin{proof}
Theorem~\ref{the:KineticSYG} gives the statements $1-3$. 

Let  $m_p$ be the number of insertions/deletions into $TT_p$. By Theorem~\ref{the:KineticTT}, all $TT_p$, for all $p\in P$, generate at most $O(\sum_pm_p\beta_{2s+2}(n)\log n)=O(\beta_{2s+2}(n)\log n\sum_pm_p)$ events. Since each edge is incident to two points, inserting (resp. deleting) an edge $pq$ into the \1syg~causes two insertions (resp. deletions) into $TT_p$ and $TT_q$. The number of all edge insertions/deletions into the \1syg~is $O(n^2\beta_{s+2}(n))$ (Theorem 6 of~\cite{Rahmati2014}, so  $\sum_pm_p=O(n^2\beta_{s+2}(n))$. Hence the number of all events by all the dynamic and kinetic tournament trees is $O(n^2\beta^2_{2s+2}(n)\log n)$, and the total cost is $O(n^2\beta^2_{2s+2}(n)\log^2 n)$. 

The ratio of the number of internal events $O(n^2\beta^2_{2s+2}(n)\log n)$ to the number of external events $O(n^2\beta_{2s})$ is polylogarithmic, which implies that the KDS is efficient. 

The ratio of the total processing time to the number of internal events that the KDS processes is polylogarithmic, and so the KDS is amortized responsive.

The total size of all the tournament trees is $O(n)$, so the number of certificates of the tournament trees is linear. Also, the number of all certificates corresponding to the kinetic sorted lists $L(u_i)$ and $L(x_l)$ is linear. Thus the KDS is compact.  Since the number of all certificates is $O(n)$, each point participates in a constant number of certificates  on average.
\end{proof}

\paragraph{The general case: any $k\geq 1$.}\label{sec:KDSforkNNs}
For maintenance of the $k$-nearest neighbors to each point $p\in P$, for any $k\geq 1$, we need to track the order of the edges incident to $p$ in the \ksyg~according to their Euclidean lengths. This can easily be done by using a kinetic sorted list. 

Let $E_p$ be the set of edges incident to point $p\in P$ in the \ksyg. Let $L(E_p)$ denote a kinetic sorted list that maintains the edges in $E_p$ according to their Euclidean lengths. The following gives the complexity of our kinetic approach.

\begin{theorem}\label{the:KinetickNNs}
For a set of $n$ moving points in $\mathbb{R}^d$, where the coordinates of each point are given by polynomials of at most constant degree $s$, our KDS for maintenance of all the $k$-nearest neighbors, ordered by distance from each point, uses $O(n\log^{d+1} n +kn)$ space and $O(n\log^{d+1} n + kn\log n)$ preprocessing time. Our KDS handles $O(n^2\phi(n))$ events, each in  amortized time $O(\log n)$.
\end{theorem}
\begin{proof}
Let $m_p$ be the number of insertions/deletions to the set $E_p$ over time. Since the cardinality of $E_p$ is $O(n)$, each insertion into a kinetic sorted list $L(E_p)$ can cause $O(n)$ swaps. Each change (\eg, inserting/deleting an edge $pq$) to the \ksyg~creates two insertions/deletions in the kinetic sorted lists $L(E_p)$ and $L(E_q)$; this implies that $\sum_p m_p=O(n\phi(n))$ (from Lemma~\ref{the:allKSYGchanges}). By Theorem~\ref{the:KineticSL}, all the kinetic sorted lists $L(E_p)$, for all $p\in P$, handle a total of $O(n\sum_p m_p)$  events, each in time $O(\log n)$. Combining with Theorem~\ref{the:KinetickSYG}, we obtain the total processing time $O(kn^2 + n\phi(n)\log^{d+1}n + n^2\phi(n)\log n) = O(n^2\phi(n)\log n)$ for all the events.
\end{proof}

Now we measure the performance of our KDS for maintenance of all the $k$-nearest neighbors in $\mathbb{R}^d$ by the four standard criteria in the KDS framework.

\begin{lemma}\label{the:kNNsPerfCri}
The efficiency, responsiveness, compactness, and locality of our KDS for maintenance of all the $k$-nearest neighbors are $O({\phi(n)\over k\beta(n/k)})$, $O(\log n)$ in an amortized sense, $O(kn)$, and $O(k)$ on average, respectively.
\end{lemma}
\begin{proof}
Fix a point $p\in P$. The distances of the points of $P\backslash \{p\}$ to $p$ create $n-1$ functions, such that each pair of them intersects at most $2s$ times. The number of changes to the (ordered) $k$-nearest neighbors $p_1,...,p_k$ of $p$ is equal to the complexity of the $(\leq k)$-level, which is $O(kn\beta(n/k))$  (by Theorem~\ref{the:k_levelComplexity}). Thus the total number of changes, for all $p\in P$, is $O(kn^2\beta(n/k))$. Since our KDS  handles $O(\phi(n) n^2)$ events (by Theorem~\ref{the:KinetickNNs}), the efficiency is $O({\phi(n)\over k\beta(n/k)})$.

Each event in our KDS can be handled in amortized time $O(\log n)$. This implies the proof of the responsiveness of the KDS.


For each two consecutive elements in each of the kinetic sorted lists $L(u_i)$, $L(x_l)$, and $L(E_p)$, we have a certificate. The size of the kinetic sorted lists $L(u_i)$ and $L(x_l)$ is $O(n)$, and the size of the kinetic sorted lists $L(E_p)$, for all $p\in P$, is $O(kn)$. This implies that the compactness of our KDS is $O(kn)$, and the number of certificates corresponding to each point is $O(k)$ on average.
\end{proof}
\subsection{\rknn~Queries for Moving Points}
Suppose we are given a query point $q\notin P$ at some time $t$. To find the reverse $k$-nearest neighbors of $q$, we seek the points in each cone $C_l(q)$ of $q$ and find ${\cal K}_l(q)$, the set of the first $k$ points in $L(P\cap C_l(q))$. The union of ${\cal K}_l(q)$, $l=0,...,c-1$, contains a set of candidate points for $q$ such that $q$ might be one of their $k$-nearest neighbors. We check whether these candidate points are the reverse $k$-nearest neighbors of $q$ at time $t$ or not; this can be easily done by application of Theorem~\ref{the:KineticAllNN}/\ref{the:KinetickNNs}, which in fact maintains the $k^{th}$ nearest neighbor $p_k$ of each $p\in P$. Note that if one asks a query at time $t$, which is coincident with the time when an event occurs in the all $k$-nearest neighbors KDS, we first handle the event and then answer the query. 

The following theorem gives the main results of this section.

\begin{theorem}\label{the:KineticRkNNQ}
Consider a set $P$ of $n$ moving points in $\mathbb{R}^d$, where the coordinates of each one are given by bounded-degree polynomials. Our KDS uses $O(n\log^{d+1} n + kn)$ space and $O(n\log^{d+1} n + kn\log n)$ preprocessing time. At any time $t$, an \rknn~query can be answered in time $O(\log^d n+k)$, and the number of reverse $k$-nearest neighbors for the query point is $O(k)$. If an event occurs at time $t$, the KDS spends polylogarithmic amortized time on updating itself.
\end{theorem}
\begin{proof}
From Lemma~\ref{the:SortingLists}, the $O(k)$ candidate points for the query point $q$ can be found in worst-case time $O(\log^d n +k)$. We use a KDS for maintenance of all the $k$-nearest neighbors over time (see  Theorem~\ref{the:KineticAllNN}/\ref{the:KinetickNNs}). Checking a candidate point can be done in $O(1)$ time by comparing distance $|pq|$ to distance $|pp_k|$; so it takes $O(k)$ time to check which of these candidate points (${\cal K}_l(q)$, $l=0,...,c-1$) are reverse $k$-nearest neighbors of the query point $q$.

If one asks a query at time $t$, which coincides with the time when one of the events in the KDS occurs, we first spend polylogarithmic amortized time to handle the event (by  Theorems~\ref{the:KineticAllNN} and~\ref{the:KinetickNNs}), and then spend worst-case time $O(\log^d n +k)$ to answer the query.
\end{proof}

\section{Kinetic All $(1+\epsilon)$-Nearest Neighbors}\label{sec:KineticEpsANN}
Let $q$ be the nearest neighbor of $p$ and let $\hat{q}$ be some point such that $|p\hat{q}|<(1+\epsilon).|pq|$. We call $\hat{q}$ the \textit{$(1+\epsilon)$-nearest neighbor} of $p$. In this section, we provide a KDS to maintain some $(1+\epsilon)$-nearest neighbor for any point $p\in P$. This KDS gives better performance than the KDS of Section~\ref{sec:app_kNNs} for maintenance of the exact all $1$-nearest neighbors.

Consider a cone $C_l$ of opening angle $\theta$, which is bounded by $d$ half-spaces. Let $x_l$ be a vector inside the cone $C_l$ that passes through the apex of $C_l$. Recall a CSPD $\Psi_{C_l}=\{(B_1,R_1),...,(B_m,R_m)\}$ for $P$ with respect to the cone $C_l$. Figure~\ref{fig:RNNgraph} depicts the cone $C_l$ and a pair $(B_i,R_i)\in \Psi_{C_l}$. 

\begin{figure}[h]
  \begin{center}
    \includegraphics[scale=1]{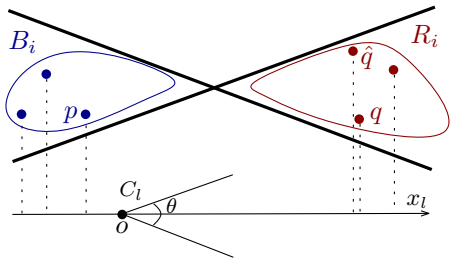}
  \end{center}
  \caption{A pair $(B_i,R_i)\in \Psi_{C_l}$.}
  \label{fig:RNNgraph}
\end{figure}

Let $b_i$ (resp. $r_i$) be the point with the maximum (resp. minimum) $x_l$-coordinate among the points in $B_i$ (resp. $R_i$). Let $E_l=\{(b_i,r_i)|~i=1,...,m\}$. We call the graph $G(P,E_l)$ the \textit{relative nearest neighbor graph} (or RNN$_l$ graph for short) with respect to $C_l$. Call the graph $G(P,\cup_l E_l)$ the \textit{RNN graph}.  The RNN graph has the following interesting properties: $(i)$ It can be constructed in $O(n\log^d n)$ time by using a $d$-dimensional RBRT, $(ii)$ it has $O(n\log^{d-1} n)$ edges, and $(iii)$ the degree of each point is $O(\log^d n)$. Lemma~\ref{the:RNNGlemma} below shows another property of the RNN graph which leads us to find some $(1+\epsilon)$-nearest neighbor for any point $p\in P$. 

\begin{lemma}\label{the:RNNGlemma}
Between all the edges incident to a point $p$ in the RNN graph, there exists an edge $(p,\hat{q})$ such that $\hat{q}$ is some $(1+\epsilon)$-nearest neighbor to $p$.
\end{lemma}
\begin{proof}
Let $q$ be the nearest neighbor to $p$ and let $q\in C_l(p)$. From the definition of a CSPD with respect to $C_l$, for $p$ and $q$ there exists a unique pair $(B_i,R_i)\in \Psi_{C_l}$ such that $p\in B_i$ and $q\in R_i$. From Lemma~\ref{the:keyLemma2}, $p$ has the maximum $x_l$-coordinate among the points in $B_i$. 

Let $\hat{q}$ be the point with the minimum $x_l$-coordinate among the points in $R_i$. For any $\epsilon>0$, there exist an appropriate angle $\theta$ and a vector $x_l$ such that $|p\hat{q}|+(1+\epsilon).|q\hat{q}|\leq (1+\epsilon).|pq|$~\cite{Abam:2011:KSX:1971362.1971367}; this satisfies that $|p\hat{q}|\leq (1+\epsilon).|pq|$. 

Therefore, the edge $(p,\hat{q})$ which is an edge of the RNN graph gives some $(1+\epsilon)$-nearest neighbor.
\end{proof}

Consider the set $E_l$ of the edges of the RNN$_l$ graph. Let $N_l(p)=\{r_i|~(b_i,r_i)\in E_l~and~b_i=p\}$. Denote by $n_l(p)$ the point in $N_l(p)$ whose $x_l$-coordinate is minimum. Let $L(N_l(p))$ be a sorted list of the points in $N_l(p)$ in ascending order according to their $x_l$-coordinates; the first point in $L(N_l(p))$ gives $n_l(p)$. 

From Lemma~\ref{the:RNNGlemma}, if the nearest neighbor of $p$ is in some set $R_i$, then $r_i$ gives some $(1+\epsilon)$-nearest neighbor to $p$. Note that we do not know which cone $C_l(p)$, $0\leq l\leq c-1$, of $p$ contains the nearest neighbor of $p$, but it is obvious that the nearest point to $p$ among these $c$ points $n_0(p),...,n_{c-1}(p)$ gives some $(1+\epsilon)$-nearest neighbor of $p$. Thus for all $l=0,...,c-1$, we track the distances of all the $n_l(p)$ to $p$ over time. A kinetic sorted list (or a tournament tree) $KSL(p)$ of size $c$ with $O(1)$ certificates can be used to maintain the nearest point to $p$. 

Similar to Section~\ref{sec:kineticKSYG} we handle two types of events, \textit{$u$-swap events} and \textit{$x$-swap events}. Note that we do not need to define a certificate for each two consecutive points in $L(N_l(.))$.  The following shows how to apply changes (\eg, insertion, deletion, and exchanging the order between two consecutive points) to the sorted lists $L(N_l(.))$ when an event occurs.

Each event can make $O(\log^d n)$ updates to the edges of $E_l$. Consider an updated pair $(b_i,r_i)$ that the value of $r_i$ (resp. $b_i$) changes from $p$ to $q$. For this update, we must delete $p$ (resp. $r_i$) form the sorted list $L(N_l(b_i))$ (resp. $L(N_l(p))$) and insert $q$ (resp. $r_i$) into $L(N_l(b_i))$ (resp. $L(N_l(q))$). If the event is an $x$-swap event, we must find all the subscripts $i$ where $r_i=q$ and check whether $n_l(b_i)=p$ or not; if so, $p$ and $q$ are in the same set $N_l(.)$ and we need to exchange their order in the corresponding sorted list $L(N_l(.))$. 

Now the following theorem gives the main result of this section. 

\begin{theorem}\label{the:KinEpsANN}
Our KDS for maintenance of all the $(1+\epsilon)$-nearest neighbors of a set of $n$ moving points in $\mathbb{R}^d$, where the trajectory of each one is an algebraic function of constant degree $s$, uses $O(n\log^{d} n)$ space and handles $O(n^2\log^d n)$ events, each in the worst-case time $O(\log^d n\log\log n)$. The KDS is compact, efficient, responsive, and local.
\end{theorem}
\begin{proof}
The proof of the preprocessing time and space follows from the properties of an RNN graph. Each event can make $O(\log^d n)$ changes to the edges of the RNN graph. Each update to a sorted list $L(N_l(.))$ can be done in $O(\log\log n)$. Thus an event can be handled in worst-case time $O(\log^dn\log\log n)$.

Since each event makes $O(\log^d n)$ changes to the values of $n_l(.)$, and since the size of each kinetic sorted list $KSL(p)$ is constant, the number of all events to maintain all the $(1+\epsilon)$-nearest neighbors is $O(n^2\log^d n)$.

Each point participates in a constant number of certificates in the kinetic sorted lists corresponding to the coordinate axes $u_i$ and $x_l$. Since the degree of each point in the RNN graph is $O(\log^d n)$, a change to the trajectory of a point may causes $O(\log^d n)$ changes in the certificates of the kinetic sorted lists $KSL(p)$. Therefore, each point participates in $O(\log^d n)$ certificates.
\end{proof}
\section{Discussion and Conclusion}\label{sec:conclusion}
We have provided KDS's for  maintenance of both the \1syg~and all the $1$-nearest neighbors, where the trajectories of the points are polynomials of degree bounded by some constant. These KDS's are amortized responsive. A future direction is to give KDS's for the \1syg~and all the $1$-nearest neighbors such that each event can be handled in a polylogarithmic worst-case time. The next open direction is to design a local KDS for maintenance of all the $1$-nearest neighbors.

Finding a linear-space KDS for all (approximate) $1$-nearest neighbors in $\mathbb{R}^d$, such that it satisfies other standard performance criteria, is an interesting future work.

In order to answer \rknn~queries over time, for any $k\geq 1$, we have provided a KDS for all the $k$-nearest neighbors. Our KDS is the first KDS for all the $k$-nearest neighbors in $\mathbb{R}^d$, for any $k\geq 1$. It processes $O(n^2\phi(n))$ events, each in amortized time $O(\log n)$. Another open problem is to design a KDS for maintenance of all the $k$-nearest neighbors that processes less than $O(n^2\phi(n))$ events.


\paragraph{\textbf{Acknowledgments}} We would like to thank Timothy M. Chan for his remarks on the best current bounds on the complexity of the $k$-level of partially-defined bounded-degree polynomials, and also for his helpful comments in the analysis of the KDS for maintenance of all the $k$-nearest neighbors.
\bibliographystyle{elsarticle-num}
\bibliography{References_Main}
\end{document}